\documentclass[runningheads,envcountsame]{llncs}

\usepackage{amsmath, amssymb}
\usepackage{graphicx}
\usepackage[nice]{nicefrac}

\usepackage{color}
\usepackage{todonotes}

\usepackage{geometry}
\spnewtheorem{observation}[theorem]{Observation}{\bfseries}{\itshape}

\newcommand{\NN}[0]{\mathbb{N}}

\newcommand{\RR}[0]{\mathbb{R}}

\newcommand{\ppdistr}[0]{\mathcal{P}}

\newcommand{\OPT}[0]{\operatorname{OPT}}
\newcommand{\DNF}[0]{\operatorname{DNF}}

\newcommand{\ew}[1]{\mathbb{E}\left[ #1 \right]}
\newcommand{\prob}[1]{\mathbb{P}\left[ #1 \right]}
\newcommand{\tv}[1]{\left\| #1 \right\|_{\operatorname{TV}}}
\newcommand{\AAPR}[2]{\operatorname{AAPR}(#1, #2)}
\newcommand{\AAPRb}[1]{\operatorname{AAPR}(#1)}
\newcommand{\AECR}[2]{\operatorname{AECR}(#1, #2)}
\newcommand{\AECRb}[1]{\operatorname{AECR}(#1)}
\newcommand{\ROR}[1]{\operatorname{RR}(#1)}
\newcommand{\RORb}[1]{\operatorname{RR}'(#1)}
\newcommand{\algo}[1]{{\scshape #1}}
\newcommand{\Unif}[0]{\operatorname{Unif}}

\allowdisplaybreaks[4]

\title{Probabilistic Analysis of the Dual Next-Fit Algorithm for Bin Covering\thanks{This research was supported by ERC Starting Grant 306465 (BeyondWorstCase).}}
\author{Carsten Fischer \and Heiko R\"oglin}
\institute{
Department of Computer Science\\
University of Bonn, Germany\\
\email{carsten.fischer@uni-bonn.de}, \email{roeglin@cs.uni-bonn.de}}

\begin{document}

\maketitle

\begin{abstract}
In the bin covering problem, the goal is to fill as many bins as possible up to a certain minimal level with a given set of items of different sizes. Online variants, in which the items arrive one after another and have to be packed immediately on their arrival without knowledge about the future items, have been studied extensively in the literature. We study the simplest possible online algorithm Dual Next-Fit, which packs all arriving items into the same bin until it is filled and then proceeds with the next bin in the same manner. The competitive ratio of this and any other reasonable online algorithm is~$1/2$.

We study Dual Next-Fit in a probabilistic setting where the item sizes are chosen i.i.d.\ according to a discrete distribution and we prove that, for every distribution, its expected competitive ratio is at least~$1/2+\epsilon$ for a constant~$\epsilon>0$ independent of the distribution. We also prove an upper bound of~$2/3$ and better lower bounds for certain restricted classes of distributions. Finally, we prove that the expected competitive ratio equals, for a large class of distributions, the random-order ratio, which is the expected competitive ratio when adversarially chosen items arrive in uniformly random order.
\end{abstract}

\section{Introduction}\label{section:Introduction}


In the \emph{bin covering problem} one is given a set of items with sizes~$s_1,\ldots,s_n\in[0,1]$ and the goal is to fill as many bins as possible with these items, where a bin is counted as filled if it contains items with a total size of at least~$1$. More precisely, we are interested in finding the maximal number~$\ell$ of pairwise disjoint sets~$X_1,\ldots,X_{\ell}\subseteq\{1,\ldots,n\}$ such that~$\sum_{i\in X_j}s_i\ge 1$ for every~$j$. We call the sets~$X_j$ \emph{bins} and we say that a bin is \emph{filled} or \emph{covered} if the total size of the items it contains is at least~$1$. Variants of the bin covering problem occur frequently in industrial applications, e.g., when packing food items with different weights into boxes that each need to have at least the advertised weight.

Bin covering is a well-studied NP-hard optimization problem. A straightforward reduction from the partition problem shows that it cannot be approximated within a factor of~$\nicefrac12+\epsilon$ for any~$\epsilon>0$. On the positive side, Jansen and Solis-Oba~\cite{JansenS03} presented an asymptotic fully polynomial-time approximation scheme. In many applications, it is natural to study online variants, in which items arrive one after another and have to be packed directly into one of the bins without knowing the future items. It is also often natural to restrict the number of open bins, i.e., bins that contain at least one item but are not yet covered, that an online algorithm may use.

We study the simple online algorithm \algo{Dual Next-Fit (DNF)} that packs all arriving items into the same bin until it is filled. Then the next items are packed into a new bin until it is filled, and so on. The (asymptotic) competitive ratio of $\DNF$ is~$\nicefrac12$~\cite{AssmannJKL84}, which is best possible for deterministic online algorithms~\cite{CsirikT88}. In fact, all deterministic online algorithms that do not add items to a bin that is already covered and have at most a constant number of open bins at any point in time have a competitive ratio of exactly~$\nicefrac12$~\cite{ChristFL14}. Since competitive analysis does not yield much insight for bin covering, alternative measures have been suggested. Most notably are probabilistic models in which the item sizes are drawn at random from a fixed distribution~\cite{CsirikFGK91} or in which the item sizes are adversarial but the items arrive in random order~\cite{ChristFL14}.

In this article, we study the \emph{asymptotic average performance ratio} and the \emph{asymptotic random-order ratio} of $\DNF$. We give now an intuitive explanation of these measures (formal definitions are given in Section~\ref{subsec:PerformanceMeasures}). In order to define the former, we allow an adversary to choose an arbitrary distribution~$F$ on~$[0,1]$ with finite support. The asymptotic average performance ratio~$\AAPR{\DNF}{F}$ of $\DNF$ with respect to the distribution~$F$ is then defined as the expected competitive ratio of $\DNF$ on instances with~$n\to\infty$ items whose sizes are independently drawn according to~$F$. Furthermore, let~$\AAPRb{\DNF}=\inf_{F}\AAPR{\DNF}{F}$. In order to define the latter, we allow an adversary to choose~$n\to\infty$ item sizes~$s_1,\ldots,s_n\in[0,1]$. The asymptotic random-order ratio~$\ROR{\DNF}$ is then defined as the expected competitive ratio of $\DNF$ on instances in which these items arrive in uniformly random order. It is assumed that the adversary chooses item sizes that minimize this expected value.

We prove several new results on the asymptotic average performance ratio and the asymptotic random-order ratio of $\DNF$ and the relation between these two measures. First of all, observe that both~$\ROR{\DNF}$ and~$\AAPRb{\DNF}$ lie between~$\nicefrac12$ and~$1$ because even in the worst case~$\DNF$ has a competitive ratio of~$\nicefrac12$. Using ideas of Kenyon~\cite{Kenyon96}, it follows that $\AAPRb{\DNF} \ge \ROR{\DNF}$. We show that~$\ROR{\DNF}\le \AAPRb{\DNF}\le \nicefrac23$, which improves a result by Christ et al.~\cite{ChristFL14} who proved that~$\ROR{\DNF} \le \nicefrac45$. To the best of our knowledge the bound by Christ et al.\ is the only non-trivial result about the random-order ratio of~$\DNF$ in the literature.

Csirik et al.~\cite{CsirikFGK91} have proved that~$\AAPR{\DNF}{F}=\nicefrac{2}{e}$ if~$F$ is the uniform distribution on~$[0,1]$. We are not aware, however, of any lower bound for~$\AAPR{\DNF}{F}$ that holds for any discrete distribution~$F$, except for the trivial bound of~$\nicefrac12$. We obtain the first such bound and prove that~$\AAPRb{\DNF}\ge \nicefrac12+\epsilon$ for a small but constant~$\epsilon>0$. We prove even better lower bounds for certain classes of distributions that we will describe in detail in Section~\ref{subsec:OurResults}. Finally we study the connection between the performance measures and prove that~$\AAPRb{\DNF} = \ROR{\DNF}$ if the adversary in the random-order model is restricted to inputs~$s_1,\ldots,s_n$ with~$\sum_is_i = \omega(n^{2/3})$.

\subsection{Performance Measures}
\label{subsec:PerformanceMeasures}

Before we discuss our results in more detail and mention further related work, let us formally introduce the performance measures that we employ.

\begin{definition}
A \emph{discrete distribution} $F$ is defined by a vector $s = (s_1, \ldots, s_m)$ of non-negative rational item sizes and an associated vector $p = (p_1, \ldots, p_m)$ of positive rational probabilities such that $\sum_{i=1}^m p_i = 1$.
\end{definition}

We denote by $I_n(F) = (X_1, \ldots, X_n)$ a list of $n$ items, where the $X_i$ are drawn i.i.d.\ according to $F$. For an algorithm~$A$ and a list of item sizes~$L$ we denote by~$A(L)$ the number of bins that~$A$ fills on input~$L$.

\begin{definition}\label{def:PerformanceMeasures}
Let $A$ be an algorithm for the bin covering problem, and let $F$ be a discrete distribution. We define the \emph{asymptotic average performance ratio} as
\begin{align*}
\AAPR{A}{F}
= \liminf_{n \rightarrow \infty} \, \ew{\frac{A(I_n(F))}{\OPT(I_n(F))}}
\end{align*}
and the \emph{asymptotic expected competitive ratio} as
\begin{align*}
\AECR{A}{F}
= \liminf_{n \rightarrow \infty} \, \frac{\ew{A(I_n(F))}}{\ew{\OPT(I_n(F))}}.
\end{align*}
For a set $\mathcal{F}$ of discrete distributions, we define
\begin{align*}
\AAPR{A}{\mathcal{F}}
= \inf_{F \in \mathcal{F}} \AAPR{A}{F}
\quad\text{and}\quad
\AECR{A}{\mathcal{F}}
= \inf_{F \in \mathcal{F}} \AECR{A}{F}.
\end{align*}
We denote by~$\mathcal{D}$ the set of all discrete distributions and we define
\[
  \AAPRb{A}=\AAPR{A}{\mathcal{D}}\quad\text{and}\quad \AECRb{A}=\AECR{A}{\mathcal{D}}.
\]
\end{definition}
Both the asymptotic average performance ratio and the asymptotic expected competitive ratio have been studied in the literature (sometimes under different names). We will see later that for our purposes there is no need to distinguish between them because they coincide for~$\DNF$.

Let $L = (a_1, \ldots, a_N)$ be a list of length $N$, and let $\sigma \in S_N$ be a permutation of $N$ elements ($S_N$ denotes the symmetric group of order $N$). Then $\sigma(L)$ denotes a permutation of $L$. 

\begin{definition}
In bin covering, the \emph{asymptotic random-order ratio} for an algorithm $A$ is defined as
\begin{align*}
\ROR{A} = \liminf_{\OPT(L) \rightarrow \infty} \frac{\mathbb{E}_{\sigma}[A(\sigma(L))]}{\OPT(L)},
\end{align*}
where $\sigma$ is drawn uniformly from $S_{|L|}$.
\end{definition}
The asymptotic random-order ratio for bin covering and bin packing has been introduced in~\cite{ChristFL14} and~\cite{Kenyon96}, respectively. All definitions above can also be adapted to the bin packing problem; we only have to replace $\inf$ and $\liminf$ by $\sup$ and $\limsup$, respectively.

\subsection{Related Work}

Csirik et al.~\cite{CsirikJK01} presented an algorithm (which requires an unlimited number of open bins) whose asymptotic average performance ratio is~$1$ for all discrete distributions. Csirik et al.~\cite{CsirikFGK91} have proved that the asymptotic expected competitive ratio of~$\DNF$ is~$\nicefrac{2}{e}$ if~$F$ is the uniform distribution on~$[0,1]$. Kenyon~\cite{Kenyon96} introduced the notion of asymptotic random-order ratio for bin packing and proved that the asymptotic random-order ratio of the best-fit algorithm lies between~$1.08$ and~$1.5$. Coffman et al.~\cite{RandomOrderBinPacking} showed that the random-order ratio of the next-fit algorithm is~$2$. Christ et al.~\cite{ChristFL14} adapted the asymptotic random-oder ratio to bin covering and proved that~$\ROR{\DNF} \le \nicefrac45$. The article of Kenyon~\cite{Kenyon96} contains in Section~3 an argument for $\AECRb{\DNF} \ge \ROR{\DNF}$ (even though this is not stated explicitly). Asgeirsson and Stein~\cite{AsgeirssonS09} developed a heuristic for online bin covering based on Markov chains and demonstrated its good performance in experiments.

\subsection{Definitions and Notation}

Let $L = (a_1, \ldots, a_N)\in[0,1]^N$ be a list of items. We denote by~$s(L) := \sum_{i=1}^N a_i$ the total size of the items in~$L$ and by~$N(L) := N$ the length of~$L$. For an algorithm~$A$, we define~$W^A(L) := s(L) - A(L)$ as the \emph{waste of algorithm~$A$ on list~$L$}. We denote by~$\OPT$ an optimal offline algorithm. Of particular interest are distributions that an optimal offline algorithm can pack with sublinear waste.

\begin{definition}\label{def:PerfectPacking}
We say that a discrete distribution $F$ is a \emph{perfect-packing distribution}, if it satisfies the perfect-packing property, i.e.,
\begin{align*}
\ew{W^{\OPT}(I_n(F))} = o(n).
\end{align*}
We denote the set of all perfect-packing distributions by~$\mathcal{P}$.
\end{definition}
Let $F$ be a discrete distribution with associated item sizes $s = (s_1, \ldots, s_m)$ and probabilities $p = (p_1, \ldots, p_m)$. We say that $b = (b_1, \ldots, b_m) \in \NN_0^m$ is a \emph{perfect-packing configuration}, if $\sum_{i=1}^m b_i s_i = 1$. Let $\Lambda_F$ denote the closure under convex combinations and positive scalar multiplication of the set of perfect-packing configurations. Courcoubetis and Weber~\cite{PerfectPackingThm} found out, that~$F$ is a perfect-packing distribution if and only if $p \in \Lambda_F$.

Let $L = (a_1, \ldots, a_N)$ be a list. We say that a discrete distribution $F$ is induced by $L$, if the vector of item sizes $(s_1, \ldots, s_m)$ contains exactly all the item sizes arising in $L$, and the vector of probabilities $p$ is given by $p_i := p(s_i) = |\{1 \leq j \leq N \, : \, a_j = s_i\}|/N$. Vice versa we can find for every discrete distribution $F$ a list $L$, such that $F$ is induced by $L$.

\subsection{Outline and Our Results}
\label{subsec:OurResults}

In Section~\ref{section:basicStatements} we discuss how~$\DNF$ can be interpreted as a Markov chain and we prove some properties using this interpretation. We investigate how the different performance measures are related and we point out that~$\AECR{\DNF}{F}=\AAPR{\DNF}{F}$ for any discrete distribution~$F$. Since the asymptotic expected competitive ratio and the asymptotic average performance ratio coincide for all discrete distributions, we will only consider the former in the following even though all mentioned results are also true for the latter. The main result of Section~\ref{section:ConnectionBetweenMeasures} is a proof that~$\AECRb{\DNF} = \ROR{\DNF}$ if the adversary in the random-order model is restricted to inputs~$s_1,\ldots,s_n$ with~$\sum_is_i = \omega(n^{2/3})$.

We start Section~\ref{section:BoundsPerfectPacking} by showing that perfect-packing distributions are the worst distributions for the considered measures, i.e., $\AECR{\DNF}{\mathcal{P}} = \AECRb{\DNF}$. Similarly we show that for proving a lower bound on the random order ratio of~$\DNF$ it suffices to consider sequences of items that can be packed with waste zero. Then we show that~$\AECRb{\DNF}\le \nicefrac23$, which implies~$\ROR{\DNF} \le \nicefrac23$. The main contribution of Section~\ref{section:BoundsPerfectPacking} are various new lower bounds on the asymptotic expected competitive ratio of~$\DNF$. We first prove that~$\AECRb{\DNF}\ge \nicefrac12+\epsilon$ for a small but constant~$\epsilon>0$. Then we consider the following special cases for which we show better lower bounds.
\begin{itemize}
  \item Let $\ppdistr_x$ be the set of all perfect packing distributions, where the maximum item size is bounded from above by $x$. For~$x\in[1/2,1]$, we prove by an application of Lorden's inequality for the overshoot of a stopped sum of random variables that
 $\AECR{\DNF}{\ppdistr_x} \geq \frac{1}{1 + x^2 + (1-x)^2}$.

\item Let $F$ be a discrete perfect-packing distribution with associated item sizes $s = (s_1, \ldots, s_m)$ and probabilities $p = (p_1, \ldots, p_m)$. According to our discussion after Definition~\ref{def:PerfectPacking}, the vector~$p$ lies in~$\Lambda_F$. Hence, there exist perfect-packing configurations~$b^1,\ldots,b^N$ and coefficients~$\alpha_1,\ldots,\alpha_N\ge 0$ with~$p=\alpha_1b^1+\ldots+\alpha_Nb^N$. We denote the smallest~$N$ for which such~$b^i$ and~$\alpha_i$ exist the \emph{degree} of~$p$. Let~$\mathcal{P}^{(N)}$ denote all discrete perfect-packing distributions with degree~$N$. We prove that
\begin{align*}
\frac{2}{3} 
\leq \AECR{\DNF}{\ppdistr^{(1)}} 
\leq \left( \sum_{i=1}^\infty \frac{(i-1)!}{i^i} \right)^{-1} 
\approx 0.736.
\end{align*}
If the maximum item size is greater than or equal to $\frac{1}{2}$ the lower bound can be improved to $\left(1 + \sum_{i=2}^\infty \frac{1}{i^2} \cdot \left(1 - \frac{1}{i} \right)^{i-2} \right)^{-1} \approx 0.686$.
\item Let $F$ be a discrete perfect-packing distribution with items $s = (s_1, \ldots, s_m)$ and probabilities $p = (p_1, \ldots, p_m)$ and let~$p=\alpha_1b^1+\ldots+\alpha_Nb^N$ for perfect-packing configurations~$b^1,\ldots,b^N$ and coefficients~$\alpha_1,\ldots,\alpha_N\ge 0$. Let~$\mathcal{P}_{\text{two}}$ denote all discrete perfect-packing distributions for which there exists such a representation in which every perfect-packing configuration~$b^i$ contains at most two non-zero entries. We show that $\AECR{\DNF}{\mathcal{P}_{\text{two}}} = \nicefrac23$.
\end{itemize}
In Section~\ref{section:Conclusions} we give some conclusions and present open problems. Appendix~\ref{section:BasicsInMC} contains some basics about Markov chains, and Appendix~\ref{section:appendixProofs} the proofs of two statements that we skipped in the main part.

\section{Basic Statements}\label{section:basicStatements}
Let $L_1$ and $L_2$ be two lists and let $L_1L_2$ denote the concatenation of them. At first, we want to point out that $\OPT$ as well as $\DNF$ are superadditive, i.e., it holds $\OPT(L_1) + \OPT(L_2) \leq \OPT(L_1L_2)$ and $\DNF(L_1) + \DNF(L_2) \leq \DNF(L_1L_2)$. Now let $F$ be a fixed discrete distribution. The limits $\gamma(F) := \lim_{n \rightarrow \infty} \ew{\OPT(I_n(F))}/n$ and $\lim_{n \rightarrow \infty} \ew{\DNF(I_n(F))}/n$ exist due to Fekete's lemma. This guarantees that the $\liminf$ in the definition of $\AECR{\DNF}{F}$ is in fact a limit.

Furthermore, the performance measures mentioned in Definition \ref{def:PerformanceMeasures} coincide in our case:
\begin{lemma}\label{lemma:PerformanceMeasuresCoincide}
Let $F$ be a discrete distribution. It holds
\begin{align*}
\AAPR{\DNF}{F} = \AECR{\DNF}{F}.
\end{align*}
\end{lemma}
A proof of a similar statement can be found in the extended version of~\cite{AverageAnalysisBP}. For our purposes it is easier to deal with $\AECR{\DNF}{F}$, so we will only mention this measure in the following.

In order to study $\ew{\DNF(I_n(F))}$, it will be useful to think of $\DNF(I_n(F))$ as a Markov chain. We will give a brief introduction to Markov chains in Section~\ref{section:BasicsInMC} in the appendix. A comprehensive overview can be found in \cite{MixingTimes}. The state space is given by the possible arising bin levels, where we subsume all bin levels greater than or equal to $1$ and the bin level $0$ to a special state, which we call the \emph{closed} state. This Markov chain is irreducible. 

Sometimes it will be necessary that the Markov chain does not start in the closed state, but with bin level $\ell$. $\DNF(\ell, L)$ denotes the number of bins that~$\DNF$ closes on input~$L$, starting with bin level $\ell$. We set $\DNF(L) := \DNF(c, L)$, where $c$ denotes the closed state.

A first important observation is that we can restrict ourselves to discrete distributions $F$, such that the Markov chain induced by $F$ and $\DNF$ is aperiodic.

\begin{lemma}\label{lemma:MarkovChain}
Let $F$ be a discrete distribution and $d \in \NN_{\geq 2}$. If the Markov chain induced by $F$ and $\DNF$ is $d$-periodic then $\AECR{\DNF}{F} = 1$.
\end{lemma}

\begin{proof}
We show that if the Markov chain, induced by $F$ and $\DNF$ is $d$-periodic with $d \geq 2$, then all item sizes lie in the interval $[d^{-1}, (d-1)^{-1})$. In this case we cannot perform better than putting $d$ arbitrary items in a bin.

Let $(s_1, \ldots, s_m)$ be the vector of item sizes corresponding to $F$. We assume that $s_1$ is the largest item size, and $s_m$ the smallest. W.l.o.g. we can assume that $s_m > 0$. Since the Markov chain is $d$-periodic, it is clear that $(d-1) \cdot s_1 < 1$, but $d \cdot s_1 \geq 1$. Therefore $s_1 \in [d^{-1}, (d-1)^{-1})$. Let $b \in \NN$, such that $(b-1)\cdot s_m <1$ and $b \cdot s_m \geq 1$. It follows from the $d$-periodicity that $b = kd$, where $k \in \NN$. Then $s_m \in [b^{-1}, (b-1)^{-1}) = [(kd)^{-1}, (kd-1)^{-1})$. Now look at the sequence consisting of $kd-2$ times item $s_m$ and finally one time item $s_1$. Then
\begin{align*}
(kd-2) \cdot \frac{1}{kd} + \frac{1}{d} \geq 1,
\end{align*}
if and only if $k \geq 2$. That means, if $k \geq 2$, there exists a sequence of items, which closes a bin using $(kd-1)$ items, which is a contradiction to $d$-periodicity. Hence, $k=1$ and $s_m \in [d^{-1}, (d-1)^{-1})$.
\qed
\end{proof}

Therefore, we will assume in the following that the Markov chain induced by a discrete distribution $F$ and $\DNF$ is aperiodic, and so it converges to a unique stationary distribution $\pi_F$. It holds $\pi_F(c) = \ew{T_{\DNF}^F}^{-1}$, where $T_{\DNF}^F$ denotes the number of items we need to close a bin, starting with bin level zero.

\begin{lemma}\label{lemma:limits}
Let $F$ be a perfect-packing distribution and $X$ be distributed according to $F$. Then
\begin{align*}
\lim_{n \rightarrow \infty} \frac{\ew{\OPT(I_n(F))}}{n} = \ew{X}.
\end{align*}
For every discrete distribution $F$, it holds
\begin{align*}
\lim_{n \rightarrow \infty} \frac{\ew{\DNF(I_n(F))}}{n} = \frac{1}{\ew{T_{\DNF}^F}}.
\end{align*}
\end{lemma}

\begin{proof}
We start proving the first equation. It holds
\begin{align*}
n \cdot \ew{X} = \ew{s(I_n(F))} = \ew{\OPT(I_n(F))} + \ew{W^{\OPT}(I_n(F))}.
\end{align*}
Since $F$ is a perfect-packing distribution, it follows that 
\begin{align*}
\left| n \cdot \ew{X} - \ew{\OPT(I_n(F))} \right| = o(n).
\end{align*}
Hence,
\begin{align*}
\lim_{n \rightarrow \infty} \frac{\ew{\OPT(I_n(F))}}{n} = \ew{X}.
\end{align*}
To prove the second equation we distinguish two cases. At first we assume that the Markov chain induced by $\DNF$ and $F$ is $d$-periodic, where $d \geq 2$. We know from the proof of Lemma \ref{lemma:MarkovChain}, that in this case $T_{\DNF}^F = d$. Therefore
$\ew{\DNF(I_n(F))} = \left\lfloor \frac{n}{d} \right\rfloor$. Hence
\begin{align*}
\lim_{n \rightarrow \infty} \frac{\ew{\DNF(I_n(F))}}{n}
= \frac{1}{d} = \frac{1}{\ew{T_{\DNF}^F}}.
\end{align*}
If the induced Markov chain is aperiodic, then
\begin{align*}
\left| \ew{\DNF(0, I_n(F))} - \mathbb{E}_{s \sim \pi_F}\left[ \DNF(s, I_n(F)) \right] \right| \leq 1.
\end{align*}
It holds $\mathbb{E}_{s \sim \pi_F}\left[ \DNF(s, I_n(F)) \right] = n \cdot \pi_F(c)$ and $\pi_F(c) = \ew{T_{\DNF}^F}^{-1}$, so
\begin{align*}
\lim_{n \rightarrow \infty} \frac{\ew{\DNF(I_n(F))}}{n}
= \lim_{n \rightarrow \infty} \frac{n \cdot \pi_F(c)}{n}
= \frac{1}{\ew{T_{\DNF}^F}}.
\end{align*}
\qed
\end{proof}

So, if $F$ is a perfect-packing distribution, it holds
\begin{align}
\label{eqn:limitAECR}
\AECR{\DNF}{F}
= \frac{1}{\ew{X} \cdot \ew{T_{\DNF}^F}}.
\end{align}

\section{Connection between Asymptotic Expected Competitive Ratio and Random-order Ratio}\label{section:ConnectionBetweenMeasures}
\label{section:ConnectionAECRandRR}

In this section we want to examine the connection between the asymptotic expected competitive ratio and the random-order ratio. At first we want to mention a result, which follows from \cite{Kenyon96}.
\begin{lemma}
\label{lemma:estimateKenyon}
It holds
\begin{align*}
\ROR{\DNF} \leq \AECRb{\DNF}.
\end{align*}
\end{lemma}

\begin{proof}
Let $L_n = \{L = (a_1, \ldots, a_n) \, : \, \prob{I_n(F) = L} > 0\}$. Then there exists a set of lists $\mathcal{L}_n$, such that $L_n = \dot{\bigcup}_{H \in \mathcal{L}_n} \{L \, : \, \exists \, \sigma \in S_n \text{ s.t. }  L = \sigma(H) \}$. Using the inequality $(\sum_{i=1}^n b_i)/(\sum_{i=1}^n c_i) \geq \min_{1 \leq i \leq n} b_i/c_i$, it follows
\begin{align*}
\frac{\ew{\DNF(I_n(F))}}{\ew{\OPT(I_n(F))}}
\geq \min_{H \in \mathcal{L}_n} \frac{\mathbb{E}_\sigma\left[\DNF(\sigma(H))\right]}{\mathbb{E}_\sigma\left[\OPT(\sigma(H))\right]}
= \min_{H \in \mathcal{L}_n} \frac{\mathbb{E}_\sigma\left[ \DNF(\sigma(H)) \right]}{\OPT(H)}.
\end{align*}
\qed
\end{proof}

We will show that the performance measures coincide if the sum of the items increases fast enough in terms of the number of items. A side product of the results of this section is the following: In \cite{randomOrderNextFit} the authors noted that in bin packing for the algorithm Next-fit and a certain list $L$, the following relationship holds:
\begin{align}
\label{eqn:motivationConcatenatedLists}
\lim_{j \rightarrow \infty} \frac{\mathbb{E}_\sigma\left[ \operatorname{NF}(\sigma(L^j)) \right]}{\OPT(L^j)}
= \AECR{\operatorname{NF}}{F},
\end{align}
where $L^j$ denotes the concatenation of $j$ copies of $L$ and $F$ is the discrete distribution induced by $L$. They asked if this result holds for arbitrary lists $L$. We can show that the answer is true in the context of $\DNF$.

In the following let $K$ denote a universal constant, which does not depend on the considered list $L$. We establish the following two bounds:

\begin{theorem}
\label{thm:differenceOptAndAsymptotic}
Let $L$ be an arbitrary instance and let $F$ be the induced discrete distribution. Then it holds
\begin{align*}
\left| \OPT(L) - N(L) \cdot \gamma(F) \right| \leq K \cdot N(L)^{2/3}.
\end{align*}
\end{theorem}

\begin{theorem}
\label{thm:differenceDnfAndAsymptotic}
Let $L$ be an arbitrary instance and let $F$ be the induced discrete distribution. We assume that the Markov chain induced by $\DNF$ and $F$ possesses a unique stationary measure $\pi_F$. Then it holds
\begin{align*}
\left| \mathbb{E}_\sigma\left[ \DNF(\sigma(L)) \right] - N(L) \cdot \pi_F(c) \right| \leq K \cdot N(L)^{2/3}.
\end{align*}
\end{theorem}

The first step of the proofs is to split up $I_{N(L)}(F)$ or $\sigma(L)$ into smaller sublists, an idea which was brought up in \cite{randomOrderNextFit}. The following lemma shows that the difference between sampling with and without replacement can be controlled if the length of the sublists is sufficiently small compared to $N(L)$.
\begin{lemma}
\label{lemma:estimateDifferenceExpectations}
Let $L = (a_1, \ldots, a_N)$, $F$ be the corresponding induced discrete distribution, and $b \in \NN$. We set $\sigma(L)_{[1:b]} = (a_{\sigma(1)}, \ldots, a_{\sigma(b)})$, where $\sigma$ is an arbitrary permutation of $L$. Then for $A \in \{\DNF, \OPT\}$ it holds
\begin{align*}
\left| \ew{A(\sigma(L)_{[1:b]})} - \ew{A(I_b(F))} \right|
\leq \frac{b^3}{N}.
\end{align*}
\end{lemma}
The proof of the lemma is based on estimates of the \emph{total variation distance}. For the sake of readability, we skip here the proof and refer to the appendix. The proofs of Theorem \ref{thm:differenceOptAndAsymptotic} and \ref{thm:differenceDnfAndAsymptotic} are somehow similar. Since the proof of the second is more compact, we present here only this one and defer the first one also to the appendix.

\begin{proof}[Theorem \ref{thm:differenceDnfAndAsymptotic}]
Let $L = (a_1, \ldots, a_N)$. For sake of readability we write $N$ instead of $N(L)$, and set $L^\sigma := \sigma(L)$. We divide $L^\sigma$ into $\left\lceil \frac{N}{b} \right\rceil$ sublists $L^\sigma_1, \ldots, L^\sigma_{\left\lceil \frac{N}{b} \right\rceil}$. Here, for $1 \leq i \leq \left\lceil \frac{N}{b} \right\rceil - 1$, it is $L^\sigma_i = (a^\sigma_{(i-1)b + 1}, \ldots, a^\sigma_{ib})$, and $L^\sigma_{\left\lceil \frac{N}{b} \right\rceil} = (a^\sigma_{\left( \left\lceil \frac{N}{b} \right\rceil - 1 \right)b + 1}, \ldots, a^\sigma_N)$. By $\mu_i$ we denote the distribution of the bin level of $\DNF$ after inserting the first $(i-1)b$ items from a random permutation of $L$.

Then it holds
\begin{multline*}
\mathbb{E}_\sigma\left[ \DNF(\sigma(L)) \right]
= \sum_{i=1}^{\left\lceil \frac{N}{b} \right\rceil} \mathbb{E}_{\ell \sim \mu_i}\left[ \DNF(\ell, L^\sigma_i) \right] \\
=  \left( \left\lceil \frac{N}{b} \right\rceil - 1 \right) \mathbb{E}_{\ell \sim \pi_F}\left[ \DNF(\ell, I_b(F)) \right] \\
+ \sum_{i=1}^{\left\lceil \frac{N}{b} \right\rceil - 1} \left( \mathbb{E}_{\ell \sim \mu_i}\left[ \DNF(\ell, L^\sigma_i) \right] - \mathbb{E}_{\ell \sim \mu_i}\left[ \DNF(\ell, I_b(F)) \right] \right) \\
+ \sum_{i=1}^{\left\lceil \frac{N}{b} \right\rceil - 1} \left( \mathbb{E}_{\ell \sim \mu_i}\left[ \DNF(\ell, I_b(F)) \right] - \mathbb{E}_{\ell \sim \pi_F}\left[ \DNF(\ell, I_b(F)) \right] \right) \\
+ \mathbb{E}_{\ell \sim \mu_{\left\lceil \frac{N}{b} \right\rceil}}\left[ \DNF(\ell, L^\sigma_i) \right].
\end{multline*}
We now have to bound the differences, which stem from different underlying probability measures and different starting points. 
According to Lemma \ref{lemma:estimateDifferenceExpectations} it is
\begin{align}
\label{eqn:estimateTV}
\left| \mathbb{E}_{\ell \sim \mu_i}\left[ \DNF(\ell, L^\sigma_i)) \right] - \mathbb{E}_{\ell \sim \mu_i}\left[ \DNF(\ell, I_b(F)) \right] \right| & \leq \frac{b^3}{N}.
\intertext{Furthermore, it is easy to show, that for $\DNF$}
\label{eqn:estimateStartingPoints}
\left| \mathbb{E}_{\ell \sim \mu_i}\left[ \DNF(\ell, I_b(F)) \right] - \mathbb{E}_{\ell \sim \pi_F}\left[ \DNF(\ell, I_b(F)) \right] \right| & \leq 1.
\end{align}
Using (\ref{eqn:estimateTV}) and (\ref{eqn:estimateStartingPoints}) we achieve the upper bound
\begin{align*}
\mathbb{E}_\sigma\left[ \DNF(\sigma(L)) \right]
& \leq \frac{N}{b} \cdot \mathbb{E}_{\ell \sim \pi_F}\left[ \DNF(\ell, I_b(F)) \right] + \frac{N}{b} \cdot \frac{b^3}{N} + \frac{N}{b} \cdot 1 + b \\
& = \frac{N}{b} \cdot b \cdot \pi_F(c) + b^2 + \frac{N}{b} + b \\
& \leq N \cdot \pi_F(c) + K \cdot N^{2/3}.
\end{align*}
In the same way, we could derive the lower bound.
\qed
\end{proof}

As a consequence of both bounds, we can give a non-trivial condition, which guarantees that the random-order ratio and the asymptotic expected competitive ratio coincide.

\begin{theorem}\label{thm:RR_AECR}
If there exists a sequence of lists $L^{(i)}$ such that
\begin{align*}
\lim_{i \rightarrow \infty} \frac{\mathbb{E}_\sigma\left[ \DNF(\sigma(L^{(i)})) \right]}{\OPT(L^{(i)})}
= \ROR{\DNF},
\end{align*}
and $s(L^{(i)}) \in \omega(N(L^{(i)})^{2/3})$, then $\ROR{\DNF} = \AECRb{\DNF}$.
\end{theorem}

\begin{proof}
Let $\epsilon > 0$ be arbitrary. Since $s(L^{(i)}) \geq \OPT(L^{(i)}) \geq \DNF(L^{(i)}) \geq s(L^{(i)})/2$, it also holds $\OPT(L^{(i)}) \in \omega(N(L^{(i)})^{2/3})$. Let $F^i$ denote the discrete distribution induced by $L^{(i)}$. Using that $\AECR{\DNF}{F^i} = \pi_{F^i}(c) / \gamma(F^i)$ and the basic inequality $|a/b - a'/b'| \leq |(a-a')/b| +  |a'/b'| \cdot |(b-b')/b|$, we obtain
\begin{align*}
\left| \frac{\mathbb{E}_\sigma\left[ \DNF(L^{(i)}) \right]}{\OPT(L^{(i)})} - \AECR{\DNF}{F^i} \right|
\leq K \cdot \frac{N(L^{(i)})^{2/3}}{\OPT(L^{(i)})}.
\end{align*}
Hence, if we choose $i$ large enough, we can find a distribution $F^i$, such that $\AECR{\DNF}{F^i} \leq \ROR{\DNF} + \epsilon$. Then
\begin{align*}
\ROR{\DNF} + \epsilon \geq \AECR{\DNF}{F^i} \geq \AECRb{\DNF} \geq \ROR{\DNF}, 
\end{align*}
i.e., both performance measures would coincide.
\qed
\end{proof}

The following corollary follows from the proof of Theorem~\ref{thm:RR_AECR}.
\begin{corollary}
Let~$\RORb{\DNF}$ denote the random-order ratio of~$\DNF$ restricted to instances~$L$ with~$s(L) \in \omega(N(L)^{2/3})$. It holds~$\RORb{\DNF} = \AECRb{\DNF}$.
\end{corollary}

Using the same method as in the proof we can also show that (\ref{eqn:motivationConcatenatedLists}) holds true for the dual next-fit algorithm.

\section{Upper and Lower Bounds for Dual Next-Fit on Perfect-packing Distributions}\label{section:BoundsPerfectPacking}
At first we show that we can restrict ourselves to studying perfect-packing distributions. They represent the worst-case with respect to the investigated performance measures.
\begin{lemma}
\label{lemma:perfectPackingIsWorstCase}
Let $L= (a_1, \ldots, a_N)$ be a list for bin covering. Then there exists a list $H$ that can be packed perfectly, i.e., $W^{\OPT}(H) = 0$, such that
\begin{align*}
\frac{\mathbb{E}_\sigma\left[ \DNF(\sigma(L)) \right]}{\OPT(L)}
\geq
\frac{\mathbb{E}_\sigma\left[ \DNF(\sigma(H)) \right]}{\OPT(H)}.
\end{align*}

Furthermore, for each distribution $F$ there exists a perfect-packing distribution $G$ such that
\begin{align*}
\lim_{n \rightarrow \infty} \frac{\ew{\DNF(I_n(F))}}{\ew{\OPT(I_n(F))}}
\geq \lim_{n \rightarrow \infty} \frac{\ew{\DNF(I_n(G))}}{\ew{\OPT(I_n(G))}}.
\end{align*}
\end{lemma}

\begin{proof}
Let $L = (a_1, \ldots, a_N)$ be a list. Then we can modify the items in such a way that we obtain a new list $H = (\tilde{a}_1, \ldots, \tilde{a}_N)$, such that $a_i \geq \tilde{a}_i$ for all $i \in [N]$, $W^{\OPT}(H) = 0$, and $\OPT(L) = \OPT(H)$. Due to the monotonicity properties of $\DNF$, it holds $\DNF(\sigma(L)) \geq \DNF(\sigma(H))$ for an arbitrary permutation $\sigma$. Hence,
\begin{align*}
\frac{\mathbb{E}_\sigma\left[ \DNF(\sigma(L)) \right]}{\OPT(L)}
\geq \frac{\mathbb{E}_\sigma\left[ \DNF(\sigma(H)) \right]}{\OPT(H)}.
\end{align*}

Now let $F$ be a discrete distribution and $L = (a_1, \ldots, a_N)$ be a list, which induces $F$ and minimizes $W^{\OPT}(L)/N(L)$. We modify $L$ in the same way as in the previous case, and let $H$ denote the modified list. We denote by $L^j$ (or $H^j$) the concatenation of $j$ copies of $L$ (or $H$), and $G$ denotes the perfect-packing distribution induced by $H$. For each $j \in \NN$ it holds $\OPT(L^j) = \OPT(H^j)$, otherwise there is a contradiction to the choice of $L$.

Furthermore, each $L^j$ induces $F$, and each $H^j$ induces $G$. Therefore, we know from Lemma \ref{lemma:Triangle1} and \ref{lemma:Triangle2}, that
\begin{align*}
\left| \ew{\OPT(I_{jN}(F))} - \ew{\OPT(I_{jN}(G))} \right|
\leq K \cdot N^{2/3} \cdot j^{2/3}.
\end{align*}
$\ew{\OPT(I_{jN}(F))}$ grows linearly in $j$, and $\DNF(I_{jN}(F)) \geq \DNF(I_{jN}(G))$. Hence,
\begin{align*}
\lim_{n \rightarrow \infty} \frac{\ew{\DNF(I_n(F))}}{\ew{\OPT(I_n(F))}}
\geq \lim_{j \rightarrow \infty} \frac{\ew{\DNF(I_{jN}(G))}}{\ew{\OPT(I_{jN}(G))}}
= \lim_{n \rightarrow \infty} \frac{\ew{\DNF(I_{n}(G))}}{\ew{\OPT(I_{n}(G))}}.
\end{align*}
\qed
\end{proof}

\subsection{Upper and Lower Bounds for Arbitrary Perfect-packing Distributions}
\label{subsection:UpperLowerBoundsArbitraryPPdistr}
We begin presenting an upper bound for the considered performance measures, which improves a result in \cite{ChristFL14}.
\begin{theorem}
\label{thm:UpperBoundPR}
It holds
\begin{align*}
\ROR{\DNF} \leq \AECRb{\DNF} \leq \frac{2}{3}.
\end{align*}
\end{theorem}

\begin{proof}
Let $F(m,k)$ be the uniform distribution on the item sizes
\begin{align*}
\left( \frac{1}{k}, 1 - \frac{1}{k}, \left(\frac{1}{k} \right)^2, 1 - \left(\frac{1}{k} \right)^2, \ldots, \left(\frac{1}{k} \right)^m, 1 - \left(\frac{1}{k} \right)^m \right).
\end{align*}
It is clear, that $F(m,k)$ is a perfect-packing distribution. We show that for every $\epsilon > 0$ there are parameters $m$ and $k$, such that $\ew{T_{\DNF}^{F(m,k)}} \geq 3 - \epsilon$.

It holds
\begin{align*}
\ew{T_{\DNF}^{F(m,k)}} 
= \sum_{i=0}^\infty \prob{T_{\DNF}^{F(m,k)} > i}
\geq 2 + \sum_{i=2}^{k-1} \prob{T_{\DNF}^{F(m,k)} > i}.
\end{align*}
Simple counting yields for $i \geq 2$
\begin{multline*}
\prob{T_{\DNF}^{F(m,k)} > i}
= \frac{m^i}{(2m)^i} + \frac{1}{(2m)^i} \sum_{j=2}^m i \cdot (j-1)^{i-1}
= \frac{1}{2^i} + \frac{i}{2^i m^i} \sum_{j=1}^{m-1} j^{i-1} \\
\geq \frac{1}{2^i} + \frac{i}{2^i m^i} \cdot \int_0^{m-1} x^{i-1} \, dx
= \frac{1}{2^i} \cdot \left[ 1 + \left( 1 - \frac{1}{m} \right)^i \right].
\end{multline*}
Therefore, if we choose at first $k$, and then $m$ large enough
\begin{align*}
\ew{T_{\DNF}^{F(m,k)}} 
\geq 2 + \sum_{i=2}^{k-1} \frac{1}{2^i} \cdot \left[ 1 + \left( 1 - \frac{1}{m} \right)^i \right]
\geq 3 - \epsilon.
\end{align*}
Using Lemma \ref{lemma:limits} the statement follows.
\qed
\end{proof}

Now we show, that in a probabilistic setting, we behave better than in the worst-case.
\begin{theorem}
\label{theorem:generalLowerBound}
There exists an $\epsilon > 0$ such that
\begin{align*}
\AECRb{\DNF} \geq \frac{1}{2} + \epsilon.
\end{align*}
\end{theorem}

\begin{proof}
Let $X_1, X_2, \ldots$ denote i.i.d.\ random variables distributed according to $F$, and $S_n = \sum_{i=1}^n X_i$. The waste, which occurs closing the bin, is given by $R = S_{T_{\DNF}^F} - 1$. We will denote $R$ also as \emph{overshoot}. Due to Wald's equation it holds that $1 + \ew{R} = \ew{S_{T_{\DNF}^F}} = \ew{T_{\DNF}^F} \cdot \ew{X}$. From (\ref{eqn:limitAECR}) it follows that $\AECR{\DNF}{F} = (1 + \ew{R})^{-1}$. We show that there exists an $\epsilon > 0$, independent of the perfect-packing distribution $F$, such that $\ew{R} \leq 1 - \epsilon$.

We assume that $F$ is induced by $\ell^\star$ perfectly packed bins and a uniform distribution on the items. Otherwise we could copy bins to achieve such a setting. We call an item \emph{large} if it is strictly larger than $1/2$. Otherwise we call the item \emph{small}. Our goal is to show that we will close a bin with a small item with a constant probability, independent of $F$.

We denote by $\ell \leq \ell^\star$ the number of large items and by $n$ the number of small items. We assume that $n \geq \ell$ and that $n$ is a multiple of $\ell$. If this is not the case, we add an appropriate number of items of size $0$. This will not change the probability of closing a bin with a small item.

Let $b_1, \ldots, b_\ell$ denote the large items and assume that $b_1 \geq b_2 \geq \ldots \geq b_\ell > 1/2$. Let $b^\star := b_{\lceil \ell/2 \rceil}$ and $s^\star = 1 - b^\star$.

Let $T$ denote the time step at which we draw for the first time a large item, and let $\mathcal{E}_F$ denote the event $\{T \geq 15n/\ell + 1\} \cap \{\sum_{i=1}^{15n/\ell} X_i \geq s^\star\}$. We will show after finishing this proof that for an arbitrary discrete distribution, it is $\prob{\mathcal{E}_F} \geq 1.1 \cdot 10^{-12}$. Let $A := \{\sum_{i=1}^{T-1} X_i < 1/2\}$. We set $q := \prob{\mathcal{E}_F^c}$, $p_1 := \prob{\mathcal{E}_F \cap A}$, and $p_2 := \prob{\mathcal{E}_F \cap A^c}$. Then, the following inequalities hold:
\begin{align*}
\ew{R} & \leq (1/2) \cdot p_2^2 + 1\cdot(1 - p_2^2) = 1 - p_2^2/2 \\
\ew{R} & \leq (1/2) \cdot p_1/2 + 1 \cdot (p_1/2 + p_2 + q) = 1 - p_1/4.
\end{align*}
The first inequality follows from the observation that~$\prob{A^c}^2\ge p_2^2$ is a lower bound on the probability that the bin gets filled with only small items, in which case the waste is at most~$1/2$. The second inequality follows because if the event~$\mathcal{E}_F \cap A$ occurs then the small items that arrive before the first large item have a total size of at least~$s^\star$ and at most~$1/2$. If the first large item has size at least~$b^\star$, which happens with probability at least~$1/2$, then it closes the bin with waste at most~$1/2$.
Since $p_1 + p_2 \geq 1.1 \cdot 10^{-12}$, it follows that $\ew{R} < 1$.
\qed
\end{proof}

In the proof we have used a lower bound for the probability of $\prob{\mathcal{E}_F}$:

\begin{lemma}
\label{lemma:estimateProbLowerBound}  
For every discrete distribution $F$ it is
\begin{align*}
\prob{\mathcal{E}_F} \geq 1.1 \cdot 10^{-12}.
\end{align*}
\end{lemma}

The interesting part in the proof of the lemma is to show that the sum of the small items exceeds the barrier $s^\star$ with positive probability. In order to prove this we reduce the summation to a kind of coupon-collectors problem. We assume that the coupons have numbers from $1$ to $n$, and we can use for an coupon with number $i$, also a coupon with a higher number. 

Let~$m\ge n$ and~$v\in\{1,\ldots,n\}^m$. We denote by~$v^{\star}$ the vector with the same entries as~$v$ except that the entries are ordered in non-increasing order. We say that~$v$ \emph{covers} the vector~$(n,\ldots,1)$ if~$v_i^{\star}\ge n-i+1$ for all~$i\in[n]$.

\begin{lemma}
Let~$a\in\{5,6,7,\ldots\}$ and let~$n\in\NN$ be arbitrary.
Let~$v$ be chosen uniformly at random from~$\{1,\ldots,n\}^{an}$. Then the probability that~$v$ does not cover the vector~$(n,\ldots,1)$ is bounded from above by
\[
    \frac{1}{e^a-1} + \frac{1}{ea}\left(\frac{1}{e^{a-2}/(2a)-1}-\frac{2a}{e^{a-2}}\right).
\]
\end{lemma}
\begin{proof}
For~$i\in[n]$, let~$U_i$ denote the event that $v^{\star}_i < n-i+1$. The event~$U_i$ occurs if and only if the number of entries from~$\{n-i+1,\ldots,n\}$ in~$v$ is at most~$i-1$. Hence,
\[
  \prob{U_1} = \prob{\forall i\in[n]:v_i<n} = \left(1-\frac{1}{n}\right)^{an} \le e^{-a}.
\]
For~$i\in \{2,\ldots,n\}$,
\begin{align*}
  \prob{U_i} & = \prob{\text{at most~$i-1$ entries of~$v$ from~$\{n-i+1,\ldots,n\}$}}\\
           & = \sum_{j=0}^{i-1} \binom{an}{j} \cdot \left(1-\frac{i}{n}\right)^{an-j}\cdot \left(\frac{i}{n}\right)^{j}\\
           & = \left(1-\frac{i}{n}\right)^{an}+\sum_{j=1}^{i-1} \binom{an}{j} \cdot \left(1-\frac{i}{n}\right)^{an-j}\cdot \left(\frac{i}{n}\right)^{j}\\
           & \le e^{-ia}+\sum_{j=1}^{i-1} \left(\frac{ean}{j}\right)^j \cdot \left(1-\frac{i}{n}\right)^{(a-1)n}\cdot \left(\frac{i}{n}\right)^{j}\\
           & = e^{-ia}+\left(1-\frac{i}{n}\right)^{(a-1)n} \cdot \sum_{j=1}^{i-1} \left(\frac{eai}{j}\right)^j\\
           & \le e^{-ia}+e^{-i(a-1)} \cdot (ea)^{i-1}\cdot \sum_{j=1}^{i-1} \left(\frac{i}{j}\right)^j\\
           & \le e^{-ia}+e^{-i(a-1)} \cdot (ea)^{i-1}\cdot \sum_{j=1}^{i-1} \binom{i}{j}\\
           & \le e^{-ia}+e^{-i(a-1)} \cdot (ea)^{i-1}\cdot 2^i.
\end{align*}

The probability that~$v$ does not cover the vector~$(n,\ldots,1)$ can be bounded from above as follows:
\begin{align*}
\prob{U_1\cup\ldots\cup U_n} & \le \sum_{i=1}^n\prob{U_i}\\
  & \le e^{-a} + \sum_{i=2}^n \left(e^{-ia}+e^{-i(a-1)} \cdot (ea)^{i-1}\cdot 2^i\right)\\
  & = \sum_{i=1}^n e^{-ia} + \frac{1}{ea}\sum_{i=2}^n \left(e^{-i(a-2)} \cdot (2a)^i\right)\\
  & \le \left(\frac{1}{1-e^{-a}}-1\right) + \frac{1}{ea}\sum_{i=2}^n \left(\frac{2a}{e^{a-2}}\right)^i\\
  & \le \frac{e^{-a}}{1-e^{-a}} + \frac{1}{ea}\left(\frac{1}{1-2a/e^{a-2}}-\frac{2a}{e^{a-2}}-1\right)\\
  & = \frac{1}{e^a-1} + \frac{1}{ea}\left(\frac{1}{e^{a-2}/(2a)-1}-\frac{2a}{e^{a-2}}\right).
\end{align*}
In the calculation we used that~$\frac{2a}{e^{a-2}} < 1$ for~$a\ge 5$.
\qed
\end{proof}

\begin{corollary}\label{cor:Cover}
The probability that a uniform random vector from~$\{1,\ldots,n\}^{5n}$ does not cover the vector~$(n,\ldots,1)$ is bounded from above by~$0.044$.
\end{corollary}

\begin{proof}[Lemma \ref{lemma:estimateProbLowerBound}]
Using that $n \geq \ell$, and so $\ell/(n+\ell) \leq 1/2$, we obtain
\begin{align*}
\prob{T \geq \frac{15n}{\ell} + 1}
= \left( 1 - \frac{\ell}{n+\ell} \right)^{15n/\ell}
\geq \left( 1 - \frac{\ell}{n+\ell} \right)^{15(n+\ell)/\ell}
\geq \frac{1}{4^{15}}.
\end{align*}
Let us condition in the following on the event that $T \geq \frac{15n}{\ell} +1$. We now want to show, that the event, that the sum of the first $15n/\ell$ small items is at least $s^\star$, occurs with positive probability independent of $F$. Thereto we partition the small items according to their size into $n/\ell$ groups with $\ell$ items each. The $\ell$ smallest items are in the first group and so on. We denote by $h$ the total size of items in the last group and $Z$ the total size of all small items. Let $v = (m_1, \ldots, m_{15n/\ell})$, where $m_i$ denotes the number of the group the $i$-th drawn item belongs to. We say that $v$ covers all groups three times, if there exists a permutation $v^\sigma$ of $v$ s.t. $(v^\sigma_1, \ldots, v^\sigma_{5n/\ell})$, $(v^\sigma_{5n/\ell + 1}, \ldots, v^\sigma_{10n/\ell})$ and $(v^\sigma_{10n/\ell + 1}, \ldots, v^\sigma_{15n/\ell})$ cover $(n,\ldots,1)$ respectively. Under the condition that $T \geq 15n/\ell + 1$, $v$ covers all groups three times with probability at least $0.956^3 \geq 0.873$ according to Corollary \ref{cor:Cover}. For $i \in [n/\ell]$ let $g_i$ denote the largest item in group $i$. If all groups are covered three times, the sum of the small items drawn is at least
\begin{align*}
3 \sum_{i=1}^{n/\ell-1} g_i \geq \frac{3(Z-h)}{\ell}.
\end{align*}
Furthermore, the total weight of all small items is at least
\begin{align*}
Z \geq (\ell - [\ell/2] + 1)s^\star \geq \ell s^\star/2.
\end{align*} 
For the following argument we can assume w.l.o.g. that there is no small item with size larger than $s^\star$ because we are only interested in the probability that the small items drawn add up to at least $s^\star$. If all groups are covered three times, the sum of the small items is at least
\begin{align*}
3 \sum_{i=1}^{n/\ell-1} g_i
\geq \frac{3(Z-h)}{\ell}
\geq \frac{3(\ell s^\star/2 - h)}{\ell}
= 3s^\star/2 - 3h/\ell.
\end{align*}
If $h \leq \ell s^\star/6$, then the sum of the small items drawn is at least $s^\star$. Hence, in this case
\begin{align*}
\prob{\mathcal{E}}
\geq 0.873 \cdot \frac{1}{4^{15}}
\geq 8 \cdot 10^{-10}.
\end{align*}
If $h > \ell s^\star/6$ then at least $\ell/11$ small items have size at least $s^\star/12$. We can see this as follows: Let $x$ denote the number of items in group $1$, which have size at least $s^\star/12$. Then $h$ is bounded from above by $xs^\star + (\ell - x)s^\star/12$. Since $h > \ell s^*/6$, it follows
\begin{align*}
(\ell s^\star)/6 < h \leq x s^\star + (\ell - x)s^\star/12.
\end{align*}
A simple computation then yields $x > \ell/11$. The probability that exactly $12$ of these items are drawn under the condition $T \geq 15n/\ell + 1$ is at least
\begin{align*}
& \binom{15n/\ell}{12}\cdot\left(1-\frac{\ell/11}{n}\right)^{15n/\ell-12}\cdot\left(\frac{\ell/11}{n}\right)^{12} \\
\geq & \left(\frac{15n}{12\ell}\right)^{12}\cdot \left(1-\frac{\ell}{11n}\right)^{15n/\ell}\cdot\left(\frac{\ell}{11n}\right)^{12}\\
\geq & \left(\frac{15}{11\cdot 12}\right)^{12}\cdot 0.239 \ge 1.1\cdot 10^{-12}.
\end{align*}
\qed
\end{proof}

\subsection{Improved Lower Bounds for Certain Classes of Perfect-packing Distributions}\label{subsection:SpecialUpperLowerBounds}
At first we look at the case that the maximum item size in the perfect-packing distribution is bounded from above by $x$. Let $\ppdistr_x$ denote the set of all such distributions.

\begin{theorem}
\label{thm:ppdistrx}
If $x \geq \frac{1}{2}$, then
\begin{align*}
\AECR{\DNF}{\ppdistr_x} \geq \frac{1}{1 + x^2 + (1-x)^2}. 
\end{align*}
\end{theorem}

The given lower bound slightly improves the worst-case bound $(1+x)^{-1}$ in the case that the maximum item size is greater than $\frac{1}{2}$.
Csirik et al.\ pointed out in \cite{CsirikFGK91} that in the case of $\DNF$ there is a connection between the bin covering problem and renewal theory, and so it is obvious to use tools from this field. The proof is based on an estimate of the overshoot, given by Lorden:
\begin{lemma}[Lorden's inequality, \cite{Lorden}]
Suppose $X_1, X_2, \ldots$ are non-negative i.i.d. random variables with $\ew{X_1} > 0$ and $\ew{X_1^2} < \infty$. Let $S_n = X_1 + \ldots + X_n$, $T = \inf\{n \in \NN \, : \, S_n \geq 1\}$, and $R = S_{T} - 1$. Then
\begin{align*}
\ew{R} \leq \ew{X_1^2}/\ew{X_1}.
\end{align*}
\end{lemma}

\begin{proof}[Theorem \ref{thm:ppdistrx}]
Let $F \in \ppdistr_x$. We have already shown in the proof of Theorem \ref{theorem:generalLowerBound} that $\AECR{\DNF}{F} = (1 + \ew{R})^{-1}$. Plugging in the estimate from Lorden's inequality yields 
\begin{align}
\label{eqn:EstimateAECRviaLorden}
\AECR{\DNF}{F} \geq \left(1 + \frac{\ew{X^2}}{\ew{X}} \right)^{-1}.
\end{align}
Let $\mathcal{C}(F)$ denote all perfect-packing configurations of $F$. Since $F$
is a perfect-packing distribution, we could write $p_F$ as $\sum_{b \in
\mathcal{C}(F)} \lambda_b \cdot c$. Therefore,
\begin{align*}
\frac{\ew{X^2}}{\ew{X}}
= \frac{\sum_{b \in \mathcal{C}(F)} \lambda_b \sum_{i=1}^m s_i^2 b_i}{\sum_{b
\in \mathcal{C}(F)} \lambda_b \sum_{i=1}^m s_i b_i}
\leq \max_{b \in \mathcal{C}(F)} \frac{\sum_{i=1}^m s_i^2 b_i}{\sum_{i=1}^m s_i
b_i}.
\end{align*}
Since $c$ are perfect-packing configurations, the denominator is equal to $1$
and the numerator is bounded from above by $x^2 + (1-x)^2$.
Plugging this into (\ref{eqn:EstimateAECRviaLorden}) yields the lower bound.
\qed
\end{proof}

Now we want to look at the case $F \in \ppdistr^{(1)}$, i.e., we have a vector of item sizes $s$, and the vector of probabilities $p_F$ is given by $p_F = b/|b|_1$, where $b$ is a perfect-packing configuration.

\begin{theorem}
\label{thm:BoundsPP1}
It holds
\begin{align*}
\frac{2}{3} 
\leq \AECR{\DNF}{\ppdistr^{(1)}} 
\leq \left( \sum_{i=1}^\infty \frac{(i-1)!}{i^i} \right)^{-1} 
\approx 0.736.
\end{align*}
If the maximum item size is greater than or equal to $\frac{1}{2}$ the lower bound can be improved to $\left(1 + \sum_{i=2}^\infty \frac{1}{i^2} \cdot \left(1 - \frac{1}{i} \right)^{i-2} \right)^{-1} \approx 0.686$.
\end{theorem}

To prove this statement, we use again the coupon-collector ansatz, already used in the proof of the lower bound for arbitrary discrete distributions.

\begin{lemma}
There are $(n+1)^{n-1}$ vectors in $\{1, \ldots, n\}^n$, which cover $(n,\ldots,1)$.
\end{lemma}

\begin{proof}
Our covering relationship is just a reformulation of \emph{parking functions}. Every vector in $\{1,\ldots,n\}^n$, which covers $(n,\ldots,1)$ is a parking function. Konheim and Weiss showed in \cite{parkingFunctions} that there are $(n+1)^{n-1}$ parking functions of length $n$.
\qed
\end{proof}

\begin{proof}[Theorem \ref{thm:BoundsPP1}]
We want to look at first at the lower bound. If the maximum item size is bounded
from above by $\frac{1}{2}$, then the lower bound is just the worst-case result: The total size of the items in one bin is bounded from above by $3/2$. Therefore, $(3/2)\cdot \DNF(L) \geq s(L) \geq \OPT(L)$. So we assume that there is an item with size greater than or equal to $1/2$.

For our construction we could assume w.l.o.g. that $F =F(m)$ is represented by $s = (s_1, \ldots, s_m)$, where $s_1 < \ldots < s_m$ and $p_F$ is the uniform distribution on $s$. We draw items i.i.d. according to $F$, and let $Y_i$ denote the index of the size of the $i$-th drawn item. Let $\tilde{T}(m)$ denote the first time, that $Y_t = m$ and $Y_u = m$, where $u < t$, or that $(Y_1, \ldots, Y_t)$ covers $(m,\ldots,1)$. In both cases the sum of the items drawn is at least $1$. It is clear that $\ew{T_{\DNF}^{F(m)}} \leq \ew{\tilde{T}(m)}$, and $\ew{\tilde{T}(1)} = 1$.

Furthermore, it holds the following recursion for $m \geq 2$:
\begin{align}
\label{recursion}
\ew{\tilde{T}(m)} = \frac{m}{m-1} \ew{\tilde{T}(m-1)} + \frac{(m-1)^{m-2}}{m^{m-1}}.
\end{align}
Let us explain this: Assume that we draw the $Y_i$ uniformly from $\{1, \ldots, m\}$, until we cover $(m, \ldots, 2)$ or draw $m$ the second time. If we stop, because we have drawn $m$ the second time, then the bin is filled. If we need more than $m-1$ items to cover $(m,\ldots,2)$, then we also cover $(m,\ldots,1)$. If we have covered $(m,\ldots,2)$ with exactly $m-1$ draws, then we have to draw one more item to cover $(m,\ldots,1)$. The previous lemma reveals that there are $(m-1)\cdot(m-2+1)^{m-1-2}$ possibilities to cover $(m,\ldots,2)$ with exactly $m-1$ items, using the largest item only once. This yields the recursion.

Applying (\ref{recursion}) several times, we achieve
\begin{multline*}
\ew{\tilde{T}(m)} 
= m \cdot \ew{\tilde{T}(1)} + \sum_{i=2}^m \frac{m}{i} \cdot \frac{(i-1)^{i-2}}{i^{i-1}} \\
= m \cdot \left( 1 + \sum_{i=2}^m \frac{1}{i^2} \left( 1 - \frac{1}{i} \right)^{i-2} \right)
\leq m \cdot \left( 1 + \sum_{i=2}^\infty \frac{1}{i^2} \left( 1 - \frac{1}{i} \right)^{i-2} \right).
\end{multline*}
Plugging this into (\ref{eqn:limitAECR}) yields the lower bound.

Now we want to show the upper bound. Assume $s_i \geq \sum_{j=i+1}^m s_j$ for $1 \leq i \leq m$. We use the same argumentation as before. This time we only need an extra item, if we draw a permutation of $(m,\ldots,2)$ in the first $m-1$ draws. This yields the recursion formula
\begin{align*}
\ew{\tilde{T}(m)} = \frac{m}{m-1} \ew{\tilde{T}(m-1)} + \frac{(m-1)!}{m^{m-1}}.
\end{align*}
Hence
\begin{align*}
\ew{\tilde{T}(m)} = m \cdot \sum_{i=1}^m \frac{(i-1)!}{i^i}.
\end{align*}
Let $s = \frac{1}{Z(m, \epsilon)}\left(1-\epsilon, (1-\epsilon)\cdot \epsilon, (1-\epsilon)\cdot \epsilon^2, \ldots, (1-\epsilon)\cdot \epsilon^{m-1} \right)$, where $Z(m, \epsilon)$ is a normalisation constant. If we choose $\epsilon$ small enough, we achieve $\ew{\tilde{T}(m)}$. This yields the upper bound.
\qed
\end{proof}

We see that even in the analysis of this simple case there is room for improvement. Based on simulations, we suppose that the upper bound represents the truth. Furthermore we were not able to improve the worst-case bound $(1+x)^{-1}$ in the case that the maximum item size is bounded from above by $\frac{1}{2}$.

Finally, let $\ppdistr_{\operatorname{two}}$ denote all discrete perfect-packing distributions for which there exists a representation in which every perfect-packing configuration $b^i$ contains at most two non-zero entries.
\begin{theorem}\label{thm:pptwo}
Let $F \in \ppdistr_{\operatorname{two}}$, then
\begin{align*}
\AECR{\DNF}{F} \geq 2/3.
\end{align*}
\end{theorem}

\begin{proof}
Since every perfect-packing configurations contains only two item sizes, we can assume that $F$ is given by the uniform distribution on the vector of item sizes $(g_1, \ldots, g_m, s_m, \ldots, s_1)$, where $g_1 \geq \ldots \geq g_m \geq s_m \geq \ldots \geq s_1$, and $g_i + s_i = 1$ for $1 \leq i \leq m$.
The items $g_i$ have size at least $1/2$. Let $\tilde{T}$ denote the first time, we have drawn two items with size at least $1/2$ or a large item and a small item, such that their sum is greater than or equal to $1/2$. We can compute $\prob{\tilde{T} > i}$ via easy counting. Using an estimate of the sum of powers of integers we achieve
\begin{multline*}
\ew{T_{\DNF}^F} 
= \sum_{i=0}^\infty \prob{T_{\DNF}^F > i}
\leq 2 + \sum_{i=2}^\infty \prob{\tilde{T} > i} \\
= 2 + \sum_{i=2}^\infty \left( \frac{m^i}{(2m)^i} + \frac{1}{(2m)^i} \sum_{j=2}^m i\cdot(j-1)^{i-1} \right) \\
= 2 + \frac{1}{2} + \sum_{i=2}^\infty \frac{1}{2^i} \cdot i \cdot  \left(\frac{1}{m^i} \sum_{j=1}^{m-1} j^{i-1}\right)
\leq \frac{5}{2} + \sum_{i=2}^\infty \frac{1}{2^i} 
= 3.
\end{multline*}
Plugging this into (\ref{eqn:limitAECR}) yields the desired bound.
\qed
\end{proof}

Combining this with the proof of Theorem \ref{thm:UpperBoundPR}, we obtain $\AECR{\DNF}{\ppdistr_{\operatorname{two}}} = 2/3$.


\section{Conclusions and Further Research}\label{section:Conclusions}

We have proven the first lower bound better than~$1/2$ for the asymptotic expected competitive ratio of~$\DNF$ that holds for any discrete distribution. Our lower bound is only slightly better than~$1/2$ and there is still a considerable gap to the best known upper bound of~$2/3$, which we also proved in this article. It is an interesting problem to close the gap between the lower and the upper bound. We conjecture that the lower bound can be improved to~$2/3$. Furthermore, we have shown that the asymptotic random-order ratio coincides with the asymptotic expected competitive ratio under the mild assumption that the adversary is not allowed to add too many too small items. We believe that this assumption is not needed and we conjecture that also for arbitrary inputs the asymptotic random-order ratio coincides with the asymptotic expected competitive ratio for~$\DNF$.

Our analysis in Section~\ref{section:ConnectionAECRandRR} that shows the connection between the asymptotic random-order ratio and the asymptotic expected competitive ratio under the previously mentioned assumption can easily be adapted to the next-fit algorithm for bin packing. We expect that it can also be generalized to more sophisticated algorithms for bin packing (e.g., to all bounded-space algorithms with only a constant number of open bins).

\bibliography{literature}

\begin{thebibliography}{10}

\bibitem{AsgeirssonS09}
Eyjolfur~Ingi Asgeirsson and Cliff Stein.
\newblock Bounded-space online bin cover.
\newblock {\em Journal of Scheduling}, 12(5):461--474, 2009.

\bibitem{AssmannJKL84}
S.~F. Assmann, David~S. Johnson, Daniel~J. Kleitman, and Joseph~Y.{-}T. Leung.
\newblock On a dual version of the one-dimensional bin packing problem.
\newblock {\em Journal of Algorithms}, 5(4):502--525, 1984.

\bibitem{ChristFL14}
Marie~G. Christ, Lene~M. Favrholdt, and Kim~S. Larsen.
\newblock Online bin covering: Expectations vs. guarantees.
\newblock {\em Theoretical Computer Science}, 556:71--84, 2014.

\bibitem{PerfectPackingThm}
Costas Courcoubetis and Richard~R. Weber.
\newblock Stability of on-line bin packing with random arrivals and long-run
  average constraints.
\newblock {\em Probability in the Engineering and Informational Sciences},
  4(4):447--460, 1990.

\bibitem{CsirikFGK91}
J{\'{a}}nos Csirik, J.~B.~G. Frenk, G{\'{a}}bor Galambos, and A.~H. G.~Rinnooy
  Kan.
\newblock Probabilistic analysis of algorithms for dual bin packing problems.
\newblock {\em Journal of Algorithms}, 12(2):189--203, 1991.

\bibitem{CsirikJK01}
J{\'{a}}nos Csirik, David~S. Johnson, and Claire Kenyon.
\newblock Better approximation algorithms for bin covering.
\newblock In {\em Proceedings of the 12th ACM-SIAM Symposium on Discrete
  Algorithms (SODA)}, pages 557--566, 2001.

\bibitem{CsirikT88}
J{\'{a}}nos Csirik and V.~Totik.
\newblock Online algorithms for a dual version of bin packing.
\newblock {\em Discrete Applied Mathematics}, 21(2):163--167, 1988.

\bibitem{randomOrderNextFit}
J~Edward G. Coffman~Jr., J{\'{a}}nos Csirik, Lajos R{\'{o}}nyai, and Ambrus
  Zsb{\'{a}}n.
\newblock Random-order bin packing.
\newblock {\em Discrete Applied Mathematics}, 156(6):2810--2816, 2008.

\bibitem{JansenS03}
Klaus Jansen and Roberto Solis{-}Oba.
\newblock An asymptotic fully polynomial time approximation scheme for bin
  covering.
\newblock {\em Theoretical Computer Science}, 306(1-3):543--551, 2003.

\bibitem{RandomOrderBinPacking}
Edward G.~Coffman Jr., J\'{a}nos Csirik, Lajos R\'{o}nyai, and Ambrus
  Zsb\'{a}n.
\newblock Random-order bin packing.
\newblock {\em Discrete Applied Mathematics}, 156(14):2810--2816, 2008.

\bibitem{Kenyon96}
Claire Kenyon.
\newblock Best-fit bin-packing with random order.
\newblock In {\em Proceedings of the 17th ACM-SIAM Symposium on Discrete
  Algorithms (SODA)}, pages 359--364, 1996.

\bibitem{parkingFunctions}
Alan~G. Konheim and Benjamin Weiss.
\newblock An occupancy discipline and applications.
\newblock {\em SIAM Journal on Applied Mathematics}, 14(6):1266--1274, 1966.

\bibitem{MixingTimes}
David~A. Levin, Yuval Peres, and Elizabeth~L. Wilmer.
\newblock {\em Markov Chains and Mixing Times}.
\newblock AMS, 2009.

\bibitem{Loh}
Wei-Liem Loh.
\newblock Stein's method and multinomial approximation.
\newblock {\em The Annals of Applied Probability}, 2(3):536--554, 1992.

\bibitem{Lorden}
Gary Lorden.
\newblock On excess over the boundary.
\newblock {\em Annals of Mathematical Statistics}, 41(2):520--527, 1970.

\bibitem{AverageAnalysisBP}
Nir Naaman and Raphael Rom.
\newblock Average case analysis of bounded space bin packing algorithms.
\newblock {\em Algorithmica}, 50:72--97, 2008.

\bibitem{OptBinCoveringRandomSizes}
Wansoo~T. Rhee and Michel Talagrand.
\newblock Optimal bin covering with items of random size.
\newblock {\em SIAM Journal on Computing}, 18(3):487--498, 1989.

\end{thebibliography}
\bibliographystyle{plain}

\newpage

\appendix
\section{Basics in Markov chains}\label{section:BasicsInMC}
A detailed introduction to the field of Markov chains can be found in \cite{MixingTimes}. We want to repeat here the basics, which are relevant for this paper.

Let $\Omega$ be a finite set and $P$ be a transition matrix on $\Omega$, i.e., for each $x \in \Omega$ $P(x, \cdot)$ is a probability distribution on $\Omega$.

A sequence of random variables $(X_0, X_1, \ldots)$ is a \emph{Markov chain} with state space $\Omega$ and transition matrix $P$, if it holds
\begin{align*}
\prob{X_{t+1} = x_{t+1} \, | \, X_{[0:t]} = x_{[0:t]}} = P(x_t, x_{t+1}),
\end{align*}
for all $t \in \NN_0$ and $x_{[0:t]} \in \Omega^{t+1}$, such that $\prob{X_{[0,t]}=x_{[0:t]}} > 0$.

Two important properties of Markov chains are \emph{irreducibility} and \emph{aperiodicity}. We say that a Markov chain is irreducible, if for all $x,y \in \Omega$ there exists an $t \in \NN$, such that $P^t(x,y) > 0$, i.e., it is possible to reach state $y$ beginning in state $x$. For each state $x \in \Omega$, we define its period as the greatest common divisor of the set $\{t \in \NN \, : \, P^t(x,x) > 0\}$. We say that a Markov chain is aperiodic, if the period of every state is equal to $1$. Otherwise, we say it is periodic. Furthermore, we can show that irreducibility implies that the periods of all states coincide.

We are especially interested in the long time behaviour of Markov chains. There, \emph{stationary distributions} play an important role. A distribution $\pi$ on $\Omega$ is said to be stationary, if it fulfils the condition $\pi = \pi P$. We will utilize the following statement: If a Markov chain is irreducible and aperiodic, then the Markov chain possesses a unique stationary measure $\pi$, and the distribution of $X_n$ converges (exponentially) to $\pi$.

Finally, there exists an interesting connection between $\pi$ and the first return time $\tau_x$. $\tau_x$ is defined as $\inf\{t \in \NN \, : \, X_t = x\}$, where $X_0 = x$. Then, it holds $\pi(x) = \mathbb{E}_x[\tau_x]^{-1}$.

\section{Proofs of statements in Section \ref{section:ConnectionAECRandRR}}\label{section:appendixProofs}
As already mentioned, $K$ denotes a universal constant, which does not depend on the considered list $L$, and could differ from line to line.
\subsection{Proof of Lemma \ref{lemma:estimateDifferenceExpectations}}
The proof of the lemma relies on the total variation distance. The total variation distance is a distance measure between two probability distributions. We will introduce the important statements for our purposes, a more thorough overview could be found in \cite{MixingTimes}.

Let $\Omega$ be a discrete set.

\begin{definition}
The \emph{total variation distance} between two probability distributions $\mu$ and $\nu$ on $\Omega$ is defined by
\begin{align*}
\tv{\mu - \nu} = \max_{A \subset \Omega} |\mu(A) - \nu(A)|.
\end{align*}
\end{definition}

We will work with two useful alternative characterizations of the total variation distance. One of them relies on couplings.

\begin{definition}
A \emph{coupling} of two probability distributions $\mu$ and $\nu$ is a pair of
random variables $(X, Y)$ defined on a single probability space such that the
marginal distribution of $X$ is $\mu$ and the marginal distribution of $Y$ is
$\nu$.
\end{definition}

\begin{lemma}[Proposition 4.2 and Proposition 4.7 in \cite{MixingTimes}]
Let $\mu$ and $\nu$ be two probability distributions on $\Omega$. Then it holds
\begin{align*}
\tv{\mu - \nu} 
& = \frac{1}{2} \sum_{x \in \Omega} | \mu(x) - \nu(x) |, \quad \text{and} \\
\tv{\mu - \nu}
& = \inf\{ \mathbb{P}[X \neq Y] \, : \, (X,Y) \text{ is a coupling of $\mu$ and
$\nu$}\}.
\end{align*}
\end{lemma}

A key ingredient is, that we are able to bound the difference between the expectations of $f(X)$ and $f(Y)$, where $X$ and $Y$ are random variables, in terms of the total variation distance and the infinity norm of $f$.

\begin{lemma}\label{lemma:estimateDifferenceViaTVandInfinityNorm}
Let $\mu, \nu$ be two probability distributions on $\Omega$, $X, Y \, : \, \Omega \rightarrow \RR$ random variables with $X \sim \mu$ and $Y \sim \nu$, and $f \, : \, \Omega \rightarrow \RR$. Then it holds
\begin{align*}
\left| \ew{f(X)} - \ew{f(Y)} \right| \leq 2 \tv{\mu - \nu} \cdot \|f\|_\infty.
\end{align*}
\end{lemma}

\begin{proof}
It holds
\begin{multline*}
\left| \ew{f(X)} - \ew{f(Y)} \right|
= \left| \sum_{x \in \Omega} (\mu(x) - \nu(x)) \cdot f(x)  \right| \\
\leq 2 \|f\|_\infty \cdot \frac{1}{2} \sum_{x \in \Omega} |\mu(x) - \nu(x)|
= 2 \|f\|_\infty \cdot \tv{\mu - \nu}.
\end{multline*}
\qed
\end{proof}

Now we want to show that the difference between sampling with and without replacement is small, if the sample size is small compared to the universe.

\begin{lemma}[Theorem 6 in \cite{Loh}]
\label{lemma:EstimateExpectationTV}
Let $L = (a_1, \ldots, a_N)$ and let $F$ be the induced discrete distribution. We assume that $\sigma \sim \Unif(S^N)$. 

Let $\mu$ be the distribution of $(a_{\sigma(1)}, \ldots, a_{\sigma(n)})$, and $\nu$ the distribution of $I_n(F)$. Then
\begin{align*}
\tv{\mu - \nu} \leq \frac{n^2}{2N}.
\end{align*}
\end{lemma}

\begin{proof}
We use the following coupling: Let $X$ be drawn according to $\nu$. We set
\begin{align*}
Y =
\begin{cases}
X, & \text{if $X$ is a legal sample w.r.t. sampling without replacement} \\
\tilde{Y}, & \tilde{Y} \sim \mu, \text{otherwise}.
\end{cases}
\end{align*}
Let $A$ be the event, that we draw $n$ different balls from an urn containing
$N$ different balls, w.r.t. sampling with replacement. Furthermore let $A_i$
denote the event, that in the $i$-th trial, we draw a ball, which was not
drawn in one of the $i-1$ trials before. Then it holds
\begin{multline*}
\tv{\mu - \nu}
\leq 1 - \prob{X = Y} \\
= 1 - \prob{X \text{ is a legal sample w.r.t. sampling without replacement}} \\
\leq 1 - \prob{A}
= 1 - \prob{\bigcap_{i=1}^n A_i}
= \prob{\bigcup_{i=1}^n A_i^c}
\leq \sum_{i=1}^n \prob{A_i^c}
= \sum_{i=1}^n \frac{i}{N}
\leq \frac{n^2}{2N}.
\end{multline*}
\qed
\end{proof}

Since $\DNF(L)$ and $\OPT(L)$ are bounded from above by $N(L)$, the proof of Lemma \ref{lemma:estimateDifferenceExpectations} follows by applying the previous lemma and Lemma \ref{lemma:estimateDifferenceViaTVandInfinityNorm}.

\subsection{Proof of Theorem \ref{thm:differenceOptAndAsymptotic}}

We divide the proof of the theorem into three lemmas.

\begin{lemma}
\label{lemma:Triangle1}
Let $L$ be an arbitrary instance, $F$ the induced discrete distribution, and $b := \lceil N(L)^{1/3} \rceil$. Then it holds
\begin{align*}
\left| \OPT(L) - \frac{N(L)}{b} \cdot \ew{\OPT(I_b(F))} \right| \leq K \cdot N(L)^{2/3}.
\end{align*}
\end{lemma}

\begin{proof}
The first step of the proof is splitting up the list into smaller sublists, and using that $\OPT$ satisfies
\begin{align*}
\OPT(L_1) + \OPT(L_2) - 1 \leq \OPT(L_1 L_2) \leq \OPT(L_1) + \OPT(L_2) + 1,
\end{align*}
where $L_1$ and $L_2$ are two lists, and $L_1 L_2$ denotes the concatenation of them. This idea was brought up in \cite{randomOrderNextFit}. Afterwards we will apply Lemma \ref{lemma:estimateDifferenceExpectations}.

Let $L = (a_1, \ldots, a_N)$. To ease notation, we will write $N$ instead of $N(L)$. Let $b = \left\lceil N^{1/3} \right\rceil$. Since $\OPT$ does not depend on the order of the items it holds $\OPT(L) = \mathbb{E}_\sigma \left[ \OPT(\sigma(L)) \right]$. Therefore we show instead
\begin{align*}
\left| \ew{\OPT(\sigma(L))} - \frac{N}{b} \cdot \ew{\OPT(I_b(F))} \right|
\leq K \cdot N^{2/3}.
\end{align*}

To simplify the notation, let $L^\sigma := \sigma(L)$. We divide $L^\sigma$ into $\left\lceil \frac{N}{b} \right\rceil$ sublists $L^\sigma_1, \ldots, L^\sigma_{\left\lceil \frac{N}{b} \right\rceil}$. Here, for $1 \leq i \leq \left\lceil \frac{N}{b} \right\rceil - 1$, it is $L^\sigma_i = (a^\sigma_{(i-1)b + 1}, \ldots, a^\sigma_{ib})$, and $L^\sigma_{\left\lceil \frac{N}{b} \right\rceil} = (a^\sigma_{\left( \left\lceil \frac{N}{b} \right\rceil - 1 \right)b + 1}, \ldots, a^\sigma_N)$. Then for $1 \leq i \leq \left\lceil \frac{N}{b} \right\rceil - 1$, $\OPT(L^\sigma_i)$ are identically distributed random variables, and $\OPT(L^\sigma_{\left\lceil \frac{N}{b} \right\rceil}) \leq b$. So, using the estimate from Lemma \ref{lemma:estimateDifferenceExpectations}, we get
\begin{multline*}
\mathbb{E}_\sigma \left[ \OPT(\sigma(L)) \right] 
\leq \left(\left\lceil \frac{N}{b} \right\rceil - 1 \right) \ew{\OPT(L^\sigma_1)} + b + \left(\left\lceil \frac{N}{b} \right\rceil - 1 \right) \\
\leq \frac{N}{b} \ew{\OPT(L^\sigma_1)} + b + \frac{N}{b}
\leq \frac{N}{b} \cdot \left( \ew{\OPT(I_b(F))} + \frac{b^3}{N} \right) + b + \frac{N}{b} \\
\leq \frac{N}{b} \ew{\OPT(I_b(F))} + b^2 + b + \frac{N}{b}.
\end{multline*}
In the same way we achieve
\begin{align*}
\mathbb{E}_\sigma \left[ \OPT(\sigma(L)) \right] \geq \frac{N}{b} \ew{\OPT(I_b(F))} - b^2 - b - \frac{N}{b},
\end{align*}
which shows the statement.
\qed
\end{proof}

\begin{lemma}
\label{lemma:Triangle2}
Let $L = (a_1, \ldots, a_N)$, $F$ be the induced discrete distribution, and $b := \lceil N(L)^{1/3} \rceil$. It holds
\begin{align*}
\left| \ew{\OPT(I_N(F))} - \frac{N(L)}{b} \cdot \ew{\OPT(I_b(F))} \right| \leq K \cdot N(L)^{2/3}.
\end{align*}
\end{lemma}

This lemma can be proved in the same way as the previous one.

\begin{lemma}
\label{lemma:Triangle3}
Let $L = (a_1, \ldots, a_N)$ be a list, and $F$ the induced discrete distribution. It holds
\begin{align*}
\left| \ew{\OPT(I_{N(L)}(F))} - N(L) \cdot \gamma(F) \right| \leq K \cdot N(L)^{2/3}.
\end{align*}
\end{lemma}
The proof of this lemma can be found in \cite{OptBinCoveringRandomSizes}.

Finally, a simple application of the triangle inequality yields the proof of the theorem.

\end{document}